\documentclass[12pt]{article}
\usepackage[utf8]{inputenc}
\usepackage[english]{babel}
\usepackage[utf8]{inputenc} % acentos sin codigo
\usepackage{float}
\usepackage{xcolor}
\usepackage{amssymb}
\usepackage{amsthm}
\usepackage{graphicx}
\usepackage{psfrag}
\usepackage{graphics}
\usepackage{amsthm}

\usepackage{afterpage}
\usepackage{lipsum}

\usepackage{caption}
\input pictex.tex

\usepackage{tikz}
\usetikzlibrary{shapes,arrows}

\theoremstyle{definition}

\newtheorem{thm}{Theorem}[section]

\newtheorem{rem}[thm]{Remark}

\begin{document}

\begin{center}
{\Large{\bf On the dynamics of the Coronavirus epidemic and the unreported cases:
the Chilean case
}}\\
\vspace{0.5cm}
{\large Andr\'es Navas\footnote{Supported by MICITEC Chile.} 
\,\,\& \,\, 
Gast\'on Vergara-Hermosilla\footnote{Supported by  the European Union’s Horizon 2020 research and 
innovation programme under the Marie Sklodowska-Curie grant agreement No 765579.}
}
\end{center}

\vspace{0.3cm}

\noindent{\bf \Large{Introduction}}

\vspace{0.3cm}

One of the main problems faced in the mathematical modeling of the coronavirus epidemic has been 
the lack of quality data. In particular, it is estimated that a large number of cases have been unreported, 
especially those of asymptomatic patients. This is mostly due to the strong demand of tests required by 
a relatively complete report of the infected cases. Although the countries/regions that have managed to 
control the epidemic have been precisely those that have been able to develop a great capacity of testing, 
this has not been achieved in most of the situations.

In general, the number of unreported patients has been estimated by extrapolating data from the reported cases 
assuming that both numbers vary proportionally. However, this view can be criticized (al least) in two directions:

\noindent -- It does not consider the dynamical role of the unreported cases in the evolution of the epidemic and, 
by extension, in the number of reported cases;

\noindent -- it does not allow a possible variation of the proportion between the numbers of reported 
and unreported cases when the conditions of social distancing remain unchanged. 

In this work, we address these points incorporating the unreported cases into the modeling. In a first introductory 
section, we discuss from a mathematical perspective what happens when the testing capacity is very low. Using 
a very simple argument we show that, in this context, the dynamics of the disease is actually governed by the growth 
of the number of unreported cases, and that of the reported patients becomes much smaller (in a progressive way) than 
that of the total cases. Although the discussion is presented in a context of extreme (hence ``ideal'') conditions, it certainly 
illustrates how essential is to consider the role of the unreported cases to analyze the global evolution of the epidemic.

Several mathematical models have been proposed to deal with the unreported cases and their role in the progression of the 
epidemic. In the second section of this work we address one of them, with the acronym SIRU, recently proposed by 
Zhihua Liu, Pierre Magal, Ousmane Seydi and Glen Webb \cite{1,2,3}. After a brief presentation of the model and 
an explanation of why it is more adjusted to the current epidemic, we proceed to establish a series of structural results. 
Although this does not completely close the study of the qualitative properties of the underlying differential equations, 
we can already glimpse an analogy with those of the classical SIR model. However, we stress an important difference: 
in the SIRU model, the curves that appear do not necessarily have a single peak. Specifically, we exhibit a simple 
method to detect parameters that give rise to curves with at least two peaks.

The model of Liu, Magal, Seydi and Webb has already been used to describe the evolution of the epidemic in various countries 
(China, South Korea, the United Kingdom, Italy, France and Spain). In the last section of this work, we implement this modeling 
to the global Chilean scenario using the official COVID-19 data provided by the Chilean government \cite {8}. However, unlike 
\cite{1,2,3}, our implementation is novel in that it uses a variable transmission rate for the disease, which is more pertinent 
according to the local epidemiological evolution.

This work concludes with a section of general conclusions in which some future lines of research are also described.

%%%%%%%%%%%%%%%%%%%%%%%%%%%%%%%%%%%%%%%%%%%%%%%%%%%%%%%%%%%%%%%%%%%%%%%%%%%

\section{On the dynamics with low testing capacity}

\subsection{The context}

We denote by $U$  the number of positive cases that one would expect to detect within a certain unit of time 
in some region using the maximal 
capacity of testing that is available.\footnote{We note that this quantity shouldn't be equal to the total number of tests, but to a 
fraction of this. Indeed, it is expected (and international statistics confirm this) that, for each positive test, a certain fixed 
amount of tests will have negative outputs. This percentage depends on numerous factors, in particular, on the bias to 
preferably test symptomatic patients.} Suppose that the epidemic has evolved to a point where the number
of detected cases in this unit time is systematically very close to this threshold $ U $. The reported 
contagion curve is then taking a plateau shape, but a question immediately arises: has the contagion process 
entered into an stationary phase -with a reproduction rate $ R $ equal to 1 or slightly higher-, or is a 
significant number of positive cases being indetected?

It is impossible to answer {\em a priori} to this question. However, it is very unlikely that
the process enters a stationary regime exactly when the threshold is reached. We argue below that, 
if the entry into an stationary regime has not occurred, then not only a fixed ``proportion'' of the number of cases is being 
undetected, but this is the case of a much larger amount; more precisely, the latter follows an exponential rate growth. 

%%%%%%%%%%%%%%%%%%%%%%%%%%%%%%%%%%%%%%%%%%%%%%%%%%%%%%%%%%%%%%%%%%%%%%%%%%

\subsection{A mathematical argument}

To simplify the discussion, our unit time will be equal to the total period of the disease.
For an instant $ i $, we denote by $ c_i $ the number of reported positive cases, and 
by $ C_i $ that of total cases. We assume that, while $ c_i $ remains below the threshold, 
it assumes values very close to it in the future evolution. We write this as follows:
$$ U-v \leq c_i \leq U, $$
where $ v $ is relatively small compared to $ U $.

We will also assume that we are not in a stationary regime (that is, the apparent stationarity is actually 
a consequence of a default of testing). Therefore, the value of $ C_i $ is substantially greater than the 
threshold and, consequently, than $ c_i $. At the instant $i$, the number of cases 
that are not being detected is $C_i - c_i$.
These individuals will have a higher social activity than those detected as infected (since the latter will be quarantined).
Therefore, the average number of new infected individuals by these undetected agents will be a value
$ \hat{\mathcal{R}} $ strictly larger than $  \mathcal{R} $.

Thus, on average, the detected infected individuals of time $i$ generate $ \mathcal{R} $ new infections at 
time $ i + 1 $, while those undetected at time $ i $ (whose number is $ C_i - c_i $) each generate 
$ \hat{\mathcal{R}}$ new infections. Therefore, the following inequality holds:
$$C_{i+1} \geq  \mathcal{R} c_i + \hat{\mathcal{R}} \, [C_i - c_i].$$
Since $ \mathcal{R} \geq 1 $, this implies
$$C_{i+1} \geq c_i + \hat{\mathcal{R}} \, [C_i-c_i],$$
hence
$$C_{i+1} - c_{i+1} \geq c_i - c_{i+1} + \hat{\mathcal{R}}  \, [C_i - c_i ] .$$
Since $ c_i \geq U - v $ and $ c_{i + 1} \leq U $, the above implies
$$C_{i+1} - c_{i+1} \geq \hat{\mathcal{R}} \, [C_i - c_i ] - v.$$
By simple recurrence, for an initial time $ i_0 $ and all $ n \geq 1 $, this gives
\begin{eqnarray*}
C_{i_0 + n} - c_{i_0 + n} 
&\geq& \hat{\mathcal{R}} \, [C_{i_0+n-1} - c_{i_0+n-1}] - v\\
&\geq& \hat{\mathcal{R}} \, [ \hat{\mathcal{R}} \, [C_{i_0+n-2} - c_{i_0+n-2}] - v] - v 
\,\, = \,\, \hat{\mathcal{R}}^2 \, [C_{i_0+n-2} - c_{i_0+n-2}] - v\, [1 + \hat{\mathcal{R}} ] \\
&\vdots& \\
&\geq& \hat{\mathcal{R}}^n \, [C_{i_0} - c_{i_0}] - v \, [1+\hat{\mathcal{R}} + \hat{\mathcal{R}}^2 + \ldots + \hat{\mathcal{R}}^{n-1}] 
\,\, =  \,\, \hat{\mathcal{R}}^n \, [C_{i_0} - c_{i_0}] - v \, \left[ \frac{\hat{R}^n - 1}{\hat{\mathcal{R}} - 1} \right]\\
&\geq& \hat{\mathcal{R}}^n \left[ C_{i_0} - c_{i_0} - \frac{v}{\hat{\mathcal{R}} - 1} \right].
\end{eqnarray*}
In other words, 
$$C_{i_0 + n}  \geq  c_{i_0 + n} + \hat{\mathcal{R}}^n \left[ C_{i_0} - c_{i_0} - \frac{v}{\hat{\mathcal{R}} - 1} \right].$$
It is natural to expect that the term $ v / (\hat{R} -1) $ is (very)
small with respect to $ C_ {i_0} - c_ {i_0} $. Indeed, on the one hand, $ \hat{R} $ does not approximate indefinitely
to $ 1 $ (since the undetected infected individuals do not significantly reduce their social activity); on the other hand, our analysis deals 
with times for which $ C_i $ is very (although, {\em a priori}, not exponentially) greater than $ c_i $. (A more robust argument 
consists in choosing not only a single initial time $ i_0 $, but to implement the previous inequality along a sequence of 
consecutive times and finally average along these inequalities.)

Assuming all of the above, the conclusion is clear: the number of infected people $ C_ {i_0 + n} $ is much higher than the number $ c_ {i_0 + n} $
of reported infected individuals. Indeed, the difference between the two is bounded from below by an exponential of ratio $ \hat{\mathcal {R}} $, while
that the evolution of the detected cases is governed by the rate $ \mathcal{R} $, which is strictly less than $ \hat {\mathcal{R}} $. So, the curve that
we see (that of the values $ c_{i + n} $) is not only very far from the real one (that of $ C_{i + n} $), but the difference between them has an  
exponential growth that we are not perceiving.

%%%%%%%%%%%%%%%%%%%%%%%%%%%%%%%%%%%%%%%%%%%%%%%%%%%%%%%%%%%%%%%%%%%%%%%%%%%

\subsection{A more theoretical discussion}

\vspace{0.2cm}

The ``toy argument'' above shows something evident: if we are not able to follow the evolution of an epidemic through appropriate mass testing, then we 
lose track of the infectious curve. In more sophisticated terms, what it reveals is that as long we are not aware that the reproduction rate $ \mathcal{R} $ is
{\em strictly smaller} than $ 1 $, the dynamical system of the epidemic moves in a regime of either slight exponential growth or, at least, of {\em instability}. 
In such a regime, small variations of the initial conditions can lead to exponential explosion. Now, to a large extent, these initial conditions are provided by 
official data. However, if these move around the maximum of what the system can detect, then we can hardly know how accurate they are and, therefore, 
whether we really are in a stationary situation or whether we have advanced to an exponential explosion of cases that we are not perceiving. For this reason, for 
each instance in which the number of positive detected cases becomes nearly constant (that is, when the curve of reported cases begins to acquire a plateau 
shape), it seems reasonable to apply a substantive increase in the number of tests (both in quantity and spectrum). If this leads to a significant increase 
in the number of positive cases, then most likely this would mean that, actually, the regime was not stationary, but was simply exceeding a threshold 
above which a significant amount of infections cannot be reported. We will return to this point in the general conclusions of this work.

%%%%%%%%%%%%%%%%%%%%%%%%%%%%%%%%%%%%%%%%%%%%%%%%%%%%%%%%%%%%%%%%%%%%%%%%%%%%%%%%

\section{About the model of Liu, Magal, Seydi and Webb}

\subsection{Presentation}

The key argument in the preceding section is that unreported patients have a greater dynamic role than those  that are reported (since the former do not 
enter into quarantine), and therefore they contribute more importantly to the epidemic. However, in the traditional SIR model, both types of patients are 
part of the same compartment. In \cite{1}, Liu, Magal, Seydi and Webb solve this problem by separating them into two compartments. Denoting respectively 
by $ S, I, R$ and $U $ the susceptible individuals, infected individuals who do not yet have symptoms (and are at incubation stage), reported infected 
individuals,  and unreported (either asymptomatic or low symptomatic) infected individuals, they consider the following diagram flux:
\begin{figure}[H]
\centering
    \includegraphics[width=0.9\textwidth]{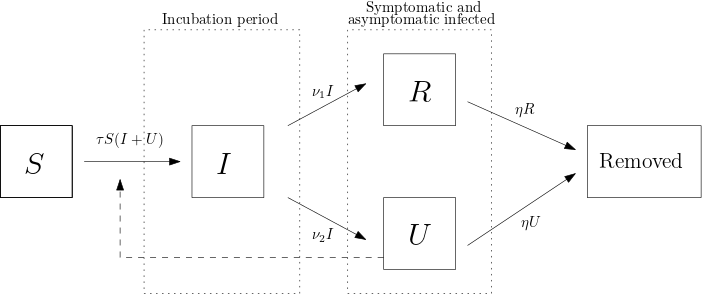}
    \caption{Diagram flux}
    \label{fig:mesh1}
\end{figure}

The differential equations attached to the diagram above and that govern the dynamics of the epidemic are the following:

\begin{equation}\label{modelo}
 \arraycolsep=1.5pt\def\arraystretch{1.5}
\left\{\begin{array}{l}S^{\prime}(t)=-\tau S(t)[I(t)+U(t)] \\ I^{\prime}(t)=\tau S(t)[I(t)+U(t)]-\nu I(t) \\ R^{\prime}(t)=\nu_{1} I(t)-\eta R(t) \\ U^{\prime}(t)=\nu_{2} I(t)-\eta U(t),\end{array}\right.
\end{equation}
where $ t \geq t_0 $ corresponds to time, with $ t_0 $ being the starting date for the study (as in \cite{1,2,3},
in the implementation, we will consider the time $ t_0 $ corresponding to the beginning of the epidemic). Although 
this system of differential equations makes perfect sense when prescribing any initial condition, in 
epidemiological modeling one is naturally lead to use data of the following type:
$$
S(t_0)=S_0>0,\quad I(t_0)=I_0>0,\quad R(t_0) \geq 0\quad \mbox{and} \quad U(t_0)=U_0\geq 0.
$$
The parameters used in the model are described in the Table below.
In particular, note that $ \nu = \nu_1 + \nu_2 $. In addition, all the parameters
that are considered $ \tau, \nu, \nu_1, \nu_2, \eta $ are positive.
\begin{table}[H]
\begin{center}
 \arraycolsep=1.4pt\def\arraystretch{1.2}
\begin{tabular}{ p{2.3cm} p{11.8cm}  }
\hline \hline
Symbol & Interpretation\\
\hline \hline
$t_0$ & Initial time.  \\
$S_0$ & Number of individuals susceptible to the disease at time $ t_0 $.  \\
$I_0$ & Number of infected individuals (in incubation period)
without symptoms at time $ t_0 $.  \\
$R_0$ & Number of reported infected individuals at time $ t_0 $. \\
$U_0$ & Number of unreported infected individuals at time $ t_0 $. \\
$\tau$ & Transmission rate of the disease.  \\
$1/\nu$ & Average time during which the infectious asymptomatic individuals remain in incubation.\\ 
$f$ & Fraction of asymptomatic infected individuals that become reported infected. \\
$\nu_{1}=f \nu$ & Rate at which asymptomatic infected cases become reported.  \\ 
$\nu_{2}=(1-f) \nu$ & Rate at which asymptomatic infected become unreported infected individuals (asymptomatic or mildly symptomatic). \\
$1/\eta$ & Average time during which an infected individual presents symptoms. \\
\hline \hline
\end{tabular}
\caption{Parameters and initial conditions of the model.}
\label{tabla:sencilla0}
\end{center}
\end{table}
\vspace{-0.5cm}

Note that in the first of the equations of (\ref{modelo}), namely
$$S' (t) = - \tau S(t) \, [I(t) + U(t)],$$
the role of $ I $ and $ U $ in the spread of the infection is the same. This is the essential point of the model: it 
gives the same dynamic role to those who are infected and do not have symptoms as to those who are not reported, 
because the latter do not go into quarantine and, in fact, have a similar social activity. Certainly, the model can be refined in many
directions; for example, one could give different though still dynamical roles to $ I $ and $ U $ by attaching to them different 
positive parameters $ \tau_1 \neq \tau_2 $ (this could be justified by that $ U $ includes individuals with low 
symptomaticity who can practice self-care). However, it is already worth to visualize some of the main 
properties of this model and to implement it in specific situations following this general  format.
 
%%%%%%%%%%%%%%%%%%%%%%%%%%%%%%%%%%%%%%%%%%%%%%%%%%%%%%%%%%%%%%%%%%%%%%%%%%%%%%%%%

\subsection{Some basic qualitative properties} 

We next consider the system of equations (\ref{modelo}) in more detail. 
We will assume an epidemic situation, which is summarized by the condition
\begin{equation}\label{epidemia}
\tau S_0 - \nu > 0.
\end{equation}
We will also consider the initial conditions of Liu, Magal, Seydi and Webb: 
\begin{equation}\label{inicio}
S_0 > 0, \,\, I_0 > 0, \,\, R_0 = 0 \,\, \mbox{ and } \,\, U_0 = \frac{\nu_2 \, I_0}{\eta + \chi_2}
\end{equation}
for a certain $ \chi_2> 0 $. (The precise value of the parameter $ \chi_2 $ will be given in the next section;
here we just retain the fact that it is positive.) For simplicity, our starting time will be  $ t_0: = 0 $.

\begin{thm}
{\em In an epidemic situation} (\ref{epidemia}) {\em and starting with the initial conditions} (\ref{inicio}), 
{\em the following three properties are fulfilled:}
\begin{enumerate}
\item[(i)] {\em For all time $t>0$, the values of  $ S (t) $, $ U (t), I (t) $ and $ R (t) $ exist, are positive and strictly
smaller than $ P: = S_0 + I_0 + U_0 $ (the total population);}
\item[(ii)]  {\em $ S (t) $ converges to a certain positive limit value $ S_{\infty} $ as $ t \to \infty $,
while $ I (t), R (t), U (t)$ converge to $ 0 $;}
\item[(iii)] {\em $ R (t) / U (t) $ converges to $ \nu_1 / \nu_2 $ from below.}
\end{enumerate}
\end{thm}

\begin{proof} We first note that, due to (\ref{epidemia}) and (\ref{inicio}), we have
$$S' (0) = -\tau S_0 I_0 < 0,$$
$$I' (0) = (\tau S_0 - \nu) I_0 + \tau S_0 U_0 > 0,$$
$$R' (0) = \nu_1 I_0 > 0,$$
$$U' (0) = \nu_2 I_0 - \eta U_0 = \nu_2 I_0 \left( 1 - \frac{\eta}{\eta + \chi_2} \right) > 0.$$
Therefore, there exists $ \varepsilon> 0 $ such that $ S (t), I (t), R (t) $ and $ U (t) $ (are defined in $ [0, \varepsilon) $ and) are 
strictly positive. The arguments that follow are inspired by an observation contained in the classical book of Vladimir Arnold \cite{arnold}.

\noindent (i) Suppose $ S $ vanishes, and let $ T > 0$ be the first time this occurs. Let $C> 0$ be such that $I(t) + U(t) \leq C$  
for all $t \in [0,T]$.  Then $S' (t) \geq - \tau C S(t)$, and hence, for $t \in [0,T)$, we have
$$\frac{S' (t)}{S(t)} \geq -\tau C .$$
Integrating between  $0$ and $s < T$, this gives
$$\log (S(s)) - \log (S_0) \geq -\tau C s,$$
and then,
$$S(s) \geq S_0 \, e^{-\tau C s}.$$
Letting $ s $ go to $ T $, this contradicts the assumption $ S (T) = 0 $.

Suppose now that $ U $ vanishes, and let $ T > 0$ be the first time this occurs. If $ I $ has not vanished until this 
time, then $U' (t) = \nu_2 I(t) - \eta U(t) \geq - \eta U(t)$ for all $t \in [0,T]$. By integration, this gives
$$U(T) \geq U_0 \, e^{-\eta T},$$
which contradicts our assumption.  Hence, $I$ should have vanished in $ [0, T] $. 

The same argument above shows that if $ R $ vanishes at a time $ T > 0$, then $ I $ must have vanished at some time in $ [0, T] $.

Finally, suppose that $ I $ vanishes, and let $ T $ be the first moment this occurs. Then, $U(t) \geq 0$ 
for all  $t \in [0,T]$, and therefore 
$$I' (t) = \tau S(t) \, [I(t)+U(t)] - \nu I(t) \geq -\nu I(t).$$
However, by integration, this again gives  a contradiction, namely $ I (T) \geq I_0 e^{- \nu T} $.

We next show that neither $ I $ nor $ U $ explode (that is, none of them tends to infinity along an increasing sequence
of times tending to a finite time $ T $). Indeed, if anyone does it in time $ T $ then, from the above,
$ P \geq S (t) \geq 0 $ on $ [0, T) $, and $ U, I $ are positive on this interval. Therefore, on $ [0, T) $,
$$(I+U)'  = \tau P \, [I+U] - \nu_1 I - \eta U \leq \tau P  \, [I + U],$$
which implies by integration that $ (I + U) (t) \leq (I_0 + U_0) \, e^{\tau P t} $ for $ t \in [0, T) $.
Letting $ t $ go to $ T $, this contradicts the explosion.

To see that $ R $ does not explode, we proceed again by contradiction: if this occurs at time $ T $, then from 
$ R = \nu_1 I - \eta R \leq \nu_1 I $ we deduce $ R (t) \leq R_0 \, e^{\nu_1 C} $ for all $ t \in [0, T) $ and $C:= \max_{t \in [0,T]} I(t)$. 

Finally, $ S $ does not explode because it is decreasing and positive.

In conclusion, $ S, I, R, U $ are all positive and do not explode. To prove that they are bounded from above by $ P $,
we introduce the equation of the {\em deletted} (removed) individuals from the system:
$$D' (t) = \eta \, [R (t) + U (t)], \qquad D(0) = 0.$$
We have a constant population $ P = S (t) + I (t) + R (t) + U (t) + D (t) $, and since $ D '> 0 $, we have that $ D ( t)> 0 $ for all $ t> 0 $.
Now, since $ S, I, R, U $ are positive for $ t> 0 $, we conclude that each of them must be strictly smaller than $ P $.

\noindent (ii) First we show that $ S $ converges towards a positive limit. Let $ S _ {\infty} $ be the limit of $ S $ (which
exists because $ S $ is decreasing). Denoting $c := \min \{\eta, \nu\}$, we have
$$(I+U)' = \tau S \, [I+U] - \nu_1 I - \eta U \leq (\tau S - c) \, (I+U).$$
Since $S' = - \tau S \, [I+U]$, this implies
$$(I+U)' \leq \frac{(\tau S - c) \, S'}{- \tau S} = - S' + \frac{c}{\tau} \frac{S'}{S}.$$
Integrating between $0$ and $t > 0$, this gives
$$(I+U+S)(t) - (I_0+U_0+S_0) \leq \frac{c}{\tau} \log \left( \frac{S(t)}{S_0} \right).$$
Since $I_0+U_0+S_0 = P$,  this implies
$$-P \leq \frac{c}{\tau} \log \left( \frac{S(t)}{S_0} \right).$$
Thus, for all $t > 0$, we have \,
$S(t) \geq S_0 e^{-\frac{\tau}{c} P},$ \,
and therefore, \, $S_{\infty} \geq S_0 e^{-\frac{\tau}{c} P} > 0.$

Next, we simultaneously prove that $ I $ and $ U $ converge to $ 0 $ (the convergence of $ R $ to $0$ will be then a 
consequence of the convergence of $ R / U $ to $ \nu_1 / \nu_2 $ proved in (iii) below). To do this, we first note that, since
$$S' = - \tau S \, [I+U],$$
it follows that there is a sequence of times $ t_n \to \infty $ such that $ (I + U) (t_n) \to 0 $. Otherwise, there would 
exist  $ c> 0 $ such that $ (I + U) (t) \geq c $ for all $ t > 0$, which implies  $ S '/ S \leq - \tau c $,  and therefore
$S(t) \leq S_0 e^{-c \tau t}$. However, for $ t $ large enough, this contradicts the inequality $S(t) \geq S_{\infty} > 0$.

Now, since $\varepsilon > 0$, we may fix $T$ such that  $(I+U) (T) \leq \varepsilon / 2$ and 
$S(t) \leq S_{\infty} + \varepsilon/2$ for all $t \geq T$. We claim that $(I+U) (t) \leq \varepsilon$ for all $t \geq T$. 
(Since $\varepsilon > 0$ was arbitrarily chosen, this concludes the proof of the convergence of $ I + U $ towards $ 0 $.) 
To prove this, note that, since $(S+I+U)' (t) = -\nu_1 I - \eta U < 0$,  we have $(S+I+U)(t) \leq (S+I+U)(T)$ for all $t \geq T$, hence
$$(I+U)(t) \leq (I+U)(T) + [S(T)-S(t)] \leq \frac{\varepsilon}{2} + \frac{\varepsilon}{2} = \varepsilon,$$
as we wanted to show.

\noindent (iii) To prove that $R/U$ converges to $\nu_1 / \nu_2$, we first remark that 
\begin{equation}\label{derivada}
\left( \frac{R}{U} \right)' = \frac{R' U - R U'}{U^2} = \frac{(\nu_1 I - \eta R ) \, U - R \, (\nu_2 I - \eta U)}{U^2} 
= \nu_2 \, \frac{I}{U} \left( \frac{\nu_1}{\nu_2} - \frac{R}{U} \right).
\end{equation}
Therefore,
$$\frac{R}{U} > \frac{\nu_1}{\nu_2} \Longrightarrow \left( \frac{R}{U} \right)' (t) < 0,$$
$$\frac{R}{U} < \frac{\nu_1}{\nu_2} \Longrightarrow \left( \frac{R}{U} \right)' (t) > 0.$$
In other words, if $ R / U $ is smaller (resp. greater) than $ \nu_1 / \nu_2 $ at a point
$ t $, then it is increasing (resp. decreasing) around this point. 

We first prove that $ R / U $ cannot be equal to $ \nu_1 / \nu_2 $ at any point. To do this, 
we note that $ R_0 / U_0 = 0 \neq \nu_1 / \nu_2 $. We define $\varphi := \frac{\nu_1}{\nu_2} - \frac{R}{U}$. 
Equality (\ref{derivada}) then becomes 
$$\varphi' = - \nu_2 \frac{I}{U} \varphi.$$
If $ T $ were the first instant at which $ R / U = \nu_1 /  \nu_2 $, then $ T $ would be  the first zero of $ \varphi $. 
By (i), there exists $C > 0$ such that $\nu_2 I / U \leq C$ on $[0,T]$. Since $R_0 = 0$, 
one has $\varphi (t) > 0$ for all positive but small-enough $t$. Choosing such a $t$ smaller than $T$ we obtain, on $[t,T),$
$$\frac{\varphi'}{\varphi} \geq - C.$$
By integration, this gives $ \varphi (T) \geq \varphi (t) \, e^{t-T} $, which contradicts the fact that $ \varphi (T) = 0 $.

To prove the convergence of $ R / U $ towards $ \nu_1 / \nu_2 $, which is equivalent to that of $ \varphi $ towards $ 0 $, 
we will use the fact (proven below) that $ I / U $ is bounded from below by a positive constant $ c$. Assuming this, we have
$$\frac{\varphi'}{\varphi} \leq -\nu_2 \, c,$$
and hence,
$$\varphi (T) \leq \varphi (t) \, e^{(t-T) \nu_2 c},$$
which converges to $ 0 $ as $T \to \infty$. 

To conclude, we must show that $ I / U $ does not approach zero. For this, we begin by noting that
$$\left( \frac{I}{U} \right)' 
= \frac{I'U - IU'}{U^2}
= \frac{(\tau S \, [I+U] - \nu I ) \, U -  I \, (\nu_2 I - \eta U)}{U^2}
= \tau S - \nu_2 \left( \frac{I}{U} \right)^2 + \frac{I}{U} (\tau S + \eta - \nu).$$
Since $ S \geq S _ {\infty} $, if $ I / U $ is very small, then the derivative $ (I / U) '$ becomes positive,
and therefore $ I / U $ grows. As a consequence, $ I / U $ cannot arbitrarily approximate $ 0 $.
\end{proof}

There are several remarks to the proof above.

\begin{rem}\label{first-r}
The proof above applies to more general initial conditions 
than (\ref{inicio}): it only requires that the values $ S_0> 0, I _0> 0, R_0 \geq 0 $ and $ U_0 \geq 0 $
are such that $ I '(t_0) $, $ R' (t_0) $ and $ U '(t_0) $ are all positive. However, note that,
at this level of generality, in statement (iii) above the convergence of $ R / U $ towards $ \nu_1 / \nu_2 $ can occur
from above, and even the quotient $ R / U $ can remain constant and equal to $ \nu_1 / \nu_2 $ throughout the whole evolution.
\end{rem}

\begin{rem} 
In the classical SIR model, the fact that the population
of susceptibles converges towards a positive limit is often presented as a consequence of the so-called 
{\em final size relation} \cite{libro}. In our context, $ S'$ not only depends on $ I $, but also on $ U $. For this reason, it is
hard to expect such a simple relation, and this partly justifies the use of robust estimates in the preceding proof.
\end{rem}

\begin{rem}
The convergence of $ R $ and $ U $ towards $0$ must hold at a speed lower than $ e^{- \eta t} $. Indeed, from 
the last two equations of the system one obtains
$$\frac{R' + \eta R}{\nu_1} = I = \frac{U' + \eta U}{\nu_2},$$
which can be rewritten in the form $\nu_2 ( R e^{\eta t})' = \nu_1 (U e^{\eta t})'$. By integration, this gives
$$\nu_2 R e^{\eta t} = \nu_1 U e^{\eta t} + C, \qquad \mbox{where} \quad C := - \nu_1 U_0.$$
Therefore,  
$$\frac{R}{U} = \frac{ \nu_1 }{ \nu_2 } + \frac{C}{\nu_2 \, U e^{\eta t}}.$$
Since $ R / U $ converges to $ \nu_1 / \nu_2 $, the product $ U e^{\eta t} $ must tend to infinity, 
thus showing our claim for $ U $. The claim for $ R $ immediately follows from this.
\end{rem}

\begin{rem}
Another observation is related to the proportion $ R / U $. Since this number grows throughout the epidemic,
the dynamic variability that we mentioned at the beginning of this work holds. Moreover, the convergence of $ R / U $ 
towards $ \nu_1 / \nu_2 $ tells us what should be the values of the parameters in the implementation. 
We will return to this point for a specific case in the next section.
\end{rem}

\begin{rem}
Finally, we must mention that one of the basic results on the SIR model has not been incorporated into the theorem 
above, namely, that of the uniqueness of the peak of the curve of infected cases. In fact, this uniqueness is no longer valid 
for the SIRU model. Below we draw an example of a double peak for the reported and unreported cases curves. Note that though this 
occurs under initial conditions different from (\ref{inicio}), it holds in a context in which the theorem is still valid, according to Remark \ref{first-r}.
\begin{figure}[H]
\centering
    \includegraphics[width=0.4\textwidth]{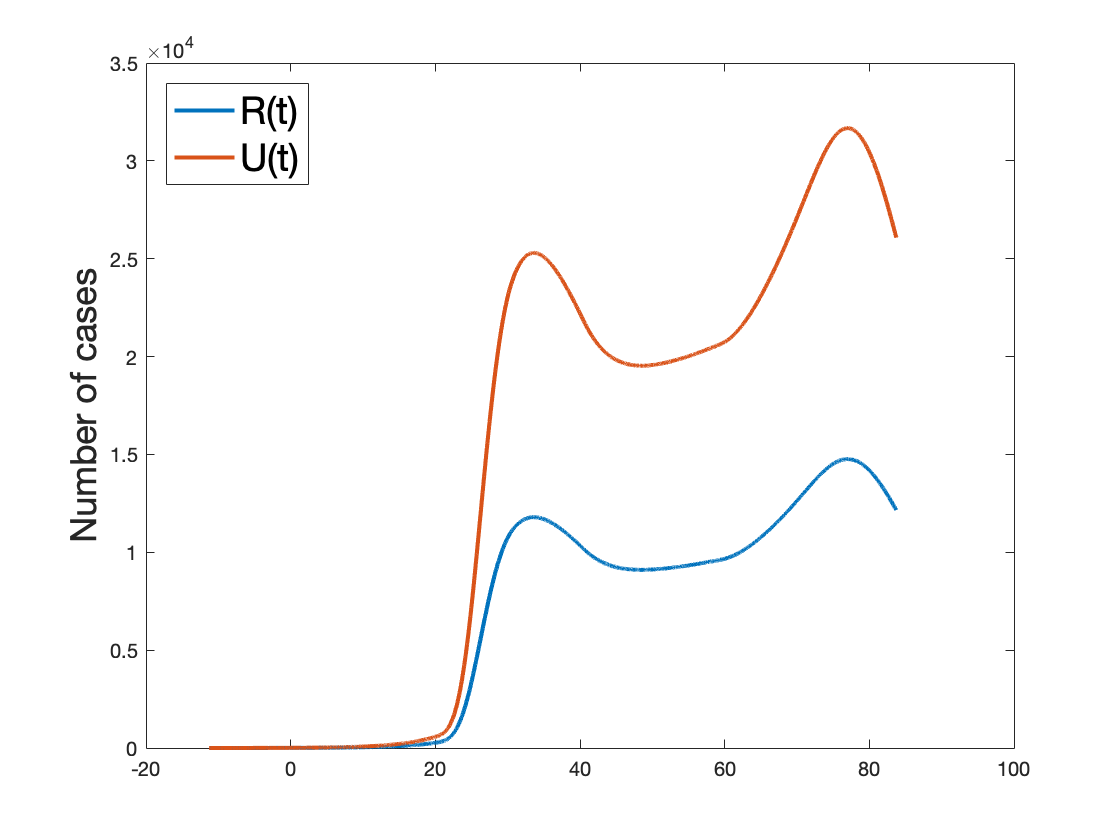}
    \caption{A double peak for the curves.}
    \label{fig:mesh2}
\end{figure}
Examples of this type can be easily obtained. For the curve $ R $, one starts by imposing the conditions
$ R '(t_p) = 0 $ and $ R''(t_p)> 0 $  for a time $ t_p > 0 $ different from the starting time. These conditions 
can be translated into conditions only on the values of $ S, I, R $ and $ U $  at that time. 
Then the SIRU equations are implemented in both directions of time around $ t_p $. This 
moment will hence correspond to a local minimum of the curve, necessarily
located between two peaks.

The same argument allows to build examples in which the curve of infected cases has two peaks. Naturally, this is due to the presence
of $ U $ in the derivative of $ I $, which is an absent element of the SIR model. In fact, in this one, the condition $ I' = 0 $ necessarily
implies $ I '' <0 $, as a straightforward computation shows. This occurs only at the peak of the curve, which corresponds to the
time when herd immunity is achieved.
\end{rem}

%%%%%%%%%%%%%%%%%%%%%%%%%%%%%%%%%%%%%%%%%%%%%%%%%%%%%%%%%%%%%%%%%%%%%%%%%

\section{Observations based on official numbers}

\subsection{General implementation of the SIRU model}
 
Our goal now is to compare the data available on COVID-19 in Chile until May 14 (2020) \cite{8} with what is 
predicted by the SIRU model, so that we can estimate the evolution of the number of unreported cases 
throughout the period. We point out that a first estimate of unreported cases in (some regions of) Chile using the
data of March 2020 and the SIRU model was carried out by M\'onica Candia and Gast\'on Vergara-Hermosilla in \cite{7}.

To begin with, let us point out that in outbreaks of influenza disease, the parameters $ \tau, \ \nu, \ \nu_1, \ \nu_2, \ \eta $, as well as the 
initial values $ S_0, I_0 $ and $ U_0 $, are generally unknown. However, it is possible to identify them from specific time data of
reported symptomatic cases. The {\em cumulative number} of infectious cases reported at a time $ t $, denoted by
$ CR (t) $, is given by
$$CR(t) := \nu_1 \int_{t_0}^t I(s)ds.$$
This is data that is openly available. Likewise, the cumulative  number of unreported cases at a time $ t $ is
$$CU(t) := \nu_2 \int_{t_0}^t I(s)ds.$$
We note that, up to constants, these values coincide respectively with
$$R(t) + \eta \int_{t_0}^t R(s) \, ds \qquad \mbox{ and } \qquad U(t) + \eta \int_{t_0}^t U(s) \, ds.$$

Following the scheme used by Liu, Magal, Seydi and Webb, we will assume that, at an early stage of the disease,
$ CR (t) $ has an (almost) exponential form:
$$
CR(t) =\chi_ 1\exp(\chi_ 2 t)-\chi_ 3.
$$
For simplicity, 
we will assume that $ \chi_3 = 1 $. The values of $ \chi_ 1 $ and $ \chi_2$ 
will then be adjusted to the accumulated data of cases reported in the early phase of the epidemic\footnote{Specifically, 
to adjust $\chi_1,\ \chi_2$ we consider the next 28 days since the detection of the first case in Chile (Talca).} using a 
classical least squares method (after passing to logarithmic coordinates). According to the above (see \cite{1} for details), 
for numerical simulations, the initial time for the beginning of the exponential growth phase is fitted at
$$
t_{0} := -\frac{1}{\chi_ {2}}\cdot\log \left(\chi_ {1}\right).
$$

Again, for simplicity, we will identify the initial value $ S_0 $ to that of the total population of Chile (since there is no prior immunity against 
the virus). Once the values of $ \nu $, $ \eta $ and $ f $ are set, the conditions at the beginning of the disease are naturally fitted as
$$
I_{0}:=\frac{\chi_ {2}}{f\left(\nu_{1}+\nu_{2}\right)} = \frac{\chi_2}{\nu_1}, 
\quad U_{0}:=\left(\frac{(1-f)\left(\nu_{1}+\nu_{2}\right)}{\eta+\chi_{2}}\right) I_{0} = \frac{\nu_2 \, I_0}{\eta + \chi_2}, 
\quad R_{0}:=0.
$$

We recall that $ 1 / \nu $ corresponds to the average time during which patients are asymptomatic infectious. As in \cite{1,2,3}, we will let this 
parameter be equal to 7 (days), and the same value will be used for the average time in which patients are reported or unreported infectious:
$$\tau = \eta = 1/7.$$ 
Note that although the incubation period has been reported as being slightly smaller, the lag in the delivery of the results of the PCR  tests 
in Chile justifies this choice.

%%%%%%%%%%%%%%%%%%%%%%%%%%%%%%%%%%%%%%%%%%%%%%%%%%%%%%%%%%%%%%%%%%%%%%%%%%%

\subsection{On the fraction of unreported cases}

To implement the SIRU system we still need to establish a good value for the parameter $ f $ (the fraction of symptomatic cases that are
reported). Once this is fitted, we will have the values of
$$
\nu_1=f\nu\quad \mbox{ y }\quad \nu_2=(1-f)\nu,
$$
and we will just need to fit the value of $ \tau $.

Before continuing, it is worth pointing out that in \cite{1,2,3} there is no major discussion on the  criterion used to establish
the value of $ f $ in the different scenarios. Actually, a value issued by the medical counterpart is assumed as valid. (For example, 
$ f =  0.8 $ is considered for China.) In our modeling, we will use the work of Baeza-Yates \cite{6} and that of Castillo and 
Past\'en \cite{CP}, who use the case fatality rate of the disease (with the correct correction according to its duration; see 
\cite{avner, letal}) to give estimates for the right number of infected individuals.\footnote{They estimate this number 
between $ 60 \% $ and $ 70 \% $ higher than the one reported for the period studied.} Since Baeza-Yates' argument is 
simpler and is not included in an academic publication, we borrow it below in a language closer to that of the SIRU  
model. As we will see, it yields a method to adjusting the value of $ f $ that can be used in almost all contexts.

Since we know that $ R / U $ converges towards $ \nu_1 / \nu_2 $, we  assume for this calculation that $ R / U $ 
is simply equal to $ \nu_1 / \nu_2 $. In addition, we will argue in discrete units of time (in days). Then we have
$$R(n) = f \, [R(n) + U(n)]$$
If $ d $ is the average time of illness to death and $ M (n) $
the number of deaths in day $n$, then the case fatality rate $ L $ corresponds to
$$L = \frac{M(n)}{[R (n-d) + U(n-d)]}.$$ 
Moreover, the reported case fatality rate is
$$L_R = \frac{M(n)}{R(n-d)} = \frac{M(n)}{[R(n-d)+U(n-d)]} \cdot \frac{[R(n-d)+U(n-d)]}{R(n-d)} = \frac{L}{f}.$$
This gives 
$$f = \frac{L}{L_R}.$$
In the Chilean context, deaths in the period studied occurred within $9.4$ days after the disease was reported.
Adjusting $ d = 9$, the computation of $ L_R $ is made feasible from the data available in \cite{8}.

Finally, regarding the natural case fatality ratio $ L $ of the disease, it is deduced from international studies that, after
adjusting it to the age distribution of the Chilean population, it should vary between $ 0.2 \% $ and $ 1 \% $, with 
a very high tendency to be around $ 0.6 \% $. In summary, this gives a parameter $ f $ varying between $ 0.1 $ 
and $ 0.5 $, with a high tendency to be close to $\frac{1}{3.4} \sim 0.3$. 

%%%%%%%%%%%%%%%%%%%%%%%%%%%%%%%%%%%%%%%%%%%%%%%%%%%%%%%%%%%%%%%%%%%%%%%%%%%

\subsection{Variations in the transmission rate} 

Given the heterogeneity of the safeguard measures taken by the Chilean government, instead of directly applying the SIRU model, it 
became more pertinent to us to consider a variable transmission rate $ \tau $ (as a function of time). To analyze the variation of the values
of $ \tau (t) $, we observe how the percentage of the Chilean population subjected to confinement has been changing, which is illustrated below:
\begin{figure}[H]
\centering
    \includegraphics[width=0.65\textwidth]{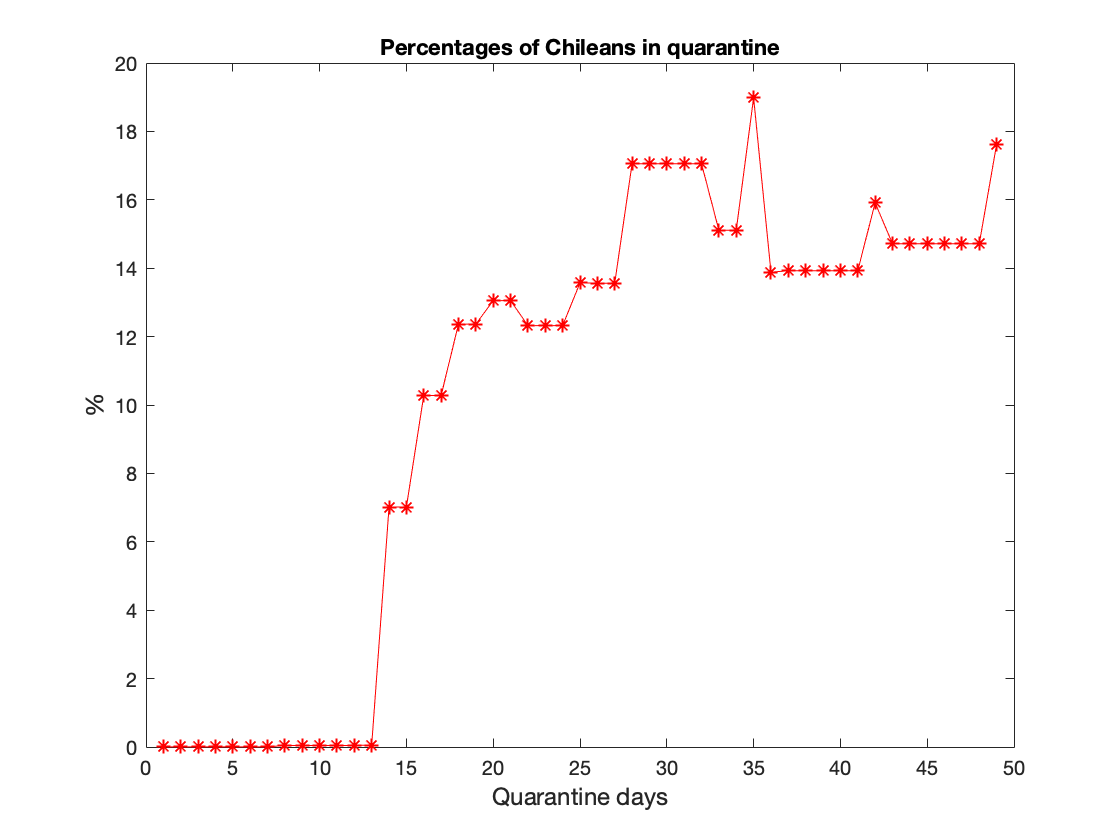}
    \caption{Variation of the percentages of the population in quarantine in Chile corresponding to the first 49 days from March 13, 2020.}
    \label{fig:mesh3}
\end{figure}
We hence propose a function $ \tau (t) $ of the form
$$
\tau(t)=\left\lbrace 
\begin{array}{ll}
\tau_0 & \mbox{ if } t\in \mathcal{I}_1 = [N_0,N_1],\\
\tau_1(t)=\tau_0 \exp(-\mu_1(t-N_1)) & \mbox{ if } t\in  \mathcal{I}_2 = (N_1,N_2],\\
\vdots& \,\, \vdots\\
\tau_r(t)=\tau_{r-1}(t) \exp(-\mu_{r-1}(t-N_{r-1})) & \mbox{ if } t\in  \mathcal{I}_{r} = (N_{r-1}, N_r],
\end{array}
\right. 
$$
where the $ \mathcal{I} _i $'s correspond to successive time intervals. For the Chilean context, according to the graph of quarantines
illustrated above (Figure 3), the chosen time intervals are those described in Table \ref{tabla:sencilla1} below:
\vspace{0.5cm}
\begin{table}[H]
\begin{center}
 \arraycolsep=1.4pt\def\arraystretch{1.1}
\begin{tabular}{ p{2.5cm} p{6cm} p{2.2cm}  }
\hline \hline
Interval   & Time frame & $N_i$\\
 \hline \hline
$\mathcal{I}_{1}$  & from March 3 to 22 & March 22 \\ 
$\mathcal{I}_{2}$  & from March 23 to April 1 & April 1  \\ 
$\mathcal{I}_3$  & from April 2 to 11 & April 11  \\ 
$ \mathcal{I}_{4}$  & from April  12 to 21  &  April 21  \\ 
$ \mathcal{I}_{5}$  & from April 22 to 31 & April 31  \\ 
$\mathcal{I}_6$  & from May 1 to 10  & May 10    \\ 
$\mathcal{I}_7$  & from May 11 to 14  & May 14  \\ 
\hline \hline
\end{tabular}
\caption{Time intervals used in our numerical simulations.}
\label{tabla:sencilla1}
\end{center}
\end{table}
\vspace{0.14cm}
Following \cite {1,2,3}, the value of the transmission rate in $ \mathcal {I} _1 $ is fitted to 
$$
\tau_0 := \left( \frac{\chi_ {2}+\nu}{S_{0}}\right) \left(  \frac{\eta+\chi_ {2}}{\nu_{2}+\eta+\chi_ {2}}\right). 
$$
Then, the parameters $ \mu_i $ are chosen in such a way that the reported cumulated cases in the numerical 
simulation align with the data of the reported cumulative number of infections at time $ t $ according to \cite{8}.
This is implemented with various values of $ f $, following the guidelines of the preceding subsection (specifically,
we work with $ f =  0.2$, $ f = 0.3$ and $ f =  0.4$). A summary is described in the following Table:
\begin{table}[H]
\begin{center}
 \arraycolsep=1.4pt\def\arraystretch{1.1}
\begin{tabular}{ p{2.5cm} p{2cm}  p{2cm}   p{2cm}   }
\hline \hline
 Parameter   &$  f=0.2$ &$  f=0.3$ & $ f=0.4$  \\
 \hline \hline
$ \mu_{1}$  & 8.3$\cdot 10^{-4}$ & 8.3$\cdot 10^{-4}$  &8.3$\cdot 10^{-4}$  \\ 
$ \mu_{2}$  & 2.1$\cdot 10^{-6}$ & 2.1$\cdot 10^{-6} $& 2.1$\cdot 10^{-6}$  \\ 
$\mu_3$  & 5$\cdot 10^{-6}$ & 5$\cdot 10^{-6}$ & 5$\cdot 10^{-6}$   \\ 
$ \mu_{4}$  & 1.41$\cdot 10^{-6}$ &1.41$\cdot 10^{-6}$ & 1.41$\cdot 10^{-6}$  \\ 
$ \mu_{5}$  & 5.08$\cdot 10^{-5}$ & 5.08$\cdot 10^{-5}$ & 5.08$\cdot 10^{-5}$  \\ 
$\mu_6$  & 2.96$\cdot 10^{-4}$ & 2.96$\cdot 10^{-4}$ & 2.96$\cdot 10^{-4}$  \\ 
\hline \hline
\end{tabular}
\caption{Parameters $ \mu_i $ corresponding to the different values of $ f $ considered in our numerical simulations.}
\label{tabla:sencilla2}
\end{center}
\end{table}
\vspace{-0.5cm}

Using the data of the cumulated reported cases available in \cite{8}, we can finally proceed to the simulations, which
are shown below. They illustrate numerical estimates for the curves of $ CR (t), \ CU (t), \ R (t) $ and $ U (t) $.

\begin{figure}[H]
\centering
    \includegraphics[width=0.71\textwidth]{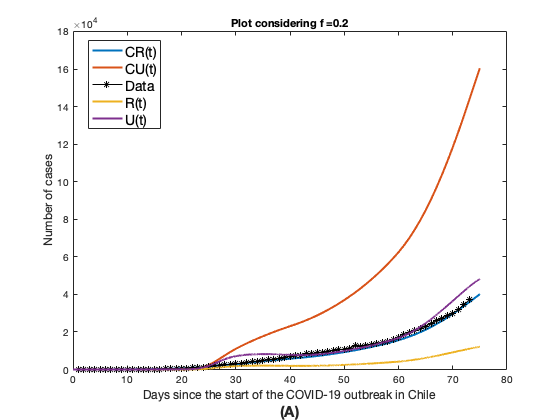}
    \centering
    \includegraphics[width=0.71\textwidth]{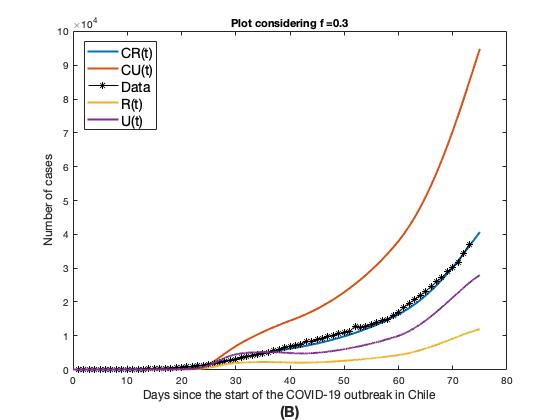}
\end{figure}
%\begin{figure}
%\centering
   % \includegraphics[width=0.85\textwidth]{03.png}
%\end{figure}
\begin{figure}[H]
    \centering
    \includegraphics[width=0.71\textwidth]{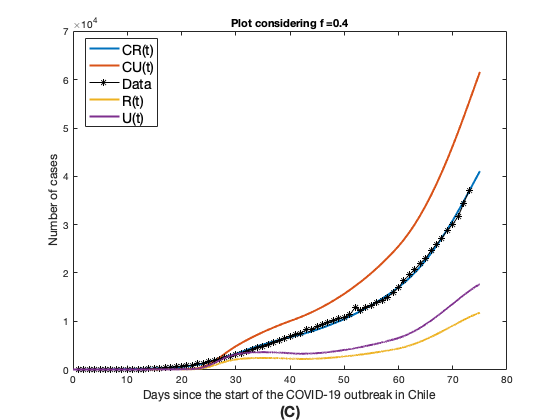}
\vspace{0.8cm}    
\caption{Plots of the numerical approximations of the functions $ CR (t), \ CU (t), \ R (t) $ and $ U (t) $ obtained from the numerical 
solutions of the model (\ref{modelo}) applied to the Chilean context based on the data of reported cumulated cases up to 14
May 2020 \cite {8}. The plots (A), (B) and (C) were obtained by considering $ f = 0.2 $, $ f = 0.3 $ and $ f = 0.4 $ respectively, 
using the parameters $\eta = 1/7$ and $\nu= 1/7$.}
    \label{fig:mesh4}
\end{figure}
%%%%%%%%%%%%%%%%%%%%
\vspace{0.5cm}

\noindent{\Large{\bf Discussion and future work}}

\vspace{0.5cm}

The epidemic outbreak of the new human coronavirus COVID-19 was first detected in Wuhan, China, in late 2019.
In Chile, the first case was reported on March 3, 2020, in Talca (Maule region). Since then, modeling the epidemic 
in the country has faced the problem of the low availability of disaggregated data \cite{baeza}.

The first part of this work was inspired by the situation experienced in Chile during the month of April, when the number of reported 
cases stabilized around 450 per day. According to most specialists, this was not accurately representing the genuine epidemiological evolution.
Although it is difficult to argue that our reasoning from the first section fully applies to this situation (in particular, because the number of daily 
tests during the period was variable), the explosive increase of cases that subsequently occurred should call us to reflection on the point. 
Nevetheless, it is clear that this increase was also due in part to the relaxation of the protection measures, which is reflected not only
in the absolute (and exponential) increase in the number of cases, but also in the proportion of positive cases with respect to  
tests (the latter despite of the significant increase in the number of tests \cite{art-ciper}).

The first part of the work naturally led us to model the dynamics of the epidemic incorporating a compartment for unreported cases so that we could deal 
with their active role in the evolution. To do this, we used the SIRU model recently introduced/implemented by Liu, Magal, Seydi and Webb in \cite{1,2,3}.  
Since the qualitative theory of the underlying differential equations has not yet been treated, we developed some essenttial points of it in the second 
section, thus rigorously establishing fundamental structure theorems. However, we pointed out an important difference with the SIR model, namely, 
the possibility of multiple peaks for the curves of infected, reported and unreported patients. For the future, it would be desirable to complete 
the qualitative description of the solutions of the equations with respect to the parameters, with an emphasis on the phenomenon of multiplicity 
of peaks. For example, until now we do not know whether there can be more than two peaks and/or whether there is any restriction on the relative 
position between them. Without any doubt, this discussion is relevant for the implementation of disease containment policies, since it is directly 
related to the recognition of the moment when the curves begin to definitively descend.

In the last section of the work, we implemented an extension of the model to the case of Chile. To do this, we considered 
a variable transmission rate, which is more appropriate according to the local reality. Our rate is coupled
to the official statistics provided by the government in \cite{8}, information that also allowed us to make the 
parametric afjustement. The last ingredient to launch the simulations was to fit the value of the fraction $f$ of infected 
individuals that were reported during the period of study. For this we used the work of Baeza-Yates \cite{6} and 
that of Castillo and Past\'en \cite{CP}. Naturally, our modeling is consistent with their estimates. In particular, in Section 2.3 of \cite{CP},
the authors estimate the number of real cases for April 28 \footnote{This corresponds to day 57 from the first case detected in Talca.} 
as 43095, which is remarkably close to the number of total cases that can be inferred from the plot (B) in Figure 4 obtained 
by considering $ f = 0.3 $. We point out that the prediction of Castillo and Past\'en is recovered 
with complete accuracy through the method used in this work 
for the value $ f = 0.31791 $. According to our modeling, on April 28, a percentage of $ 31.791 \% $ of the 
infections was reported to health agencies, and consequently $ 68.209 \% $ of the infected cases was not reported.

Although the estimates above were obtained {\em a posteriori}, their complementarity puts us in a good position  
to use them for modeling the future evolution of the epidemic. However, working in this direction requires great caution.
In particular, it would be necessary to consider a variable parameter $f$, since both the quantity and the criteria of the tests have been
modified during late May. Despite this, we strongly believe that the basis for pursuing the implementation of the SIRU model are fullfilled, 
and it would be very useful to advance in a more compartmentalized and georeferenced implementation of it.

%%%%%%%%%%%%%%%%%%%%

\vspace{0.2cm}

\noindent{\bf Acknowledgments.} We would like to thank Mar\'ia Paz Bertoglia, M\'onica Candia, Alicia Dickenstein, Yamileth Granizo,
Rafael Labarca,  Mario Ponce and Marius Tucsnak for their reading, their kind remarks and suggestions of bibliography.

%%%%%%%%%%%%%%%%%%%%

\vspace{-0.3cm}

\begin{small}

\vspace{0.1cm}

\noindent Andr\'es Navas\\

\vspace{-0.2cm}

\noindent Dpto. de Matem\'aticas y Ciencias de la Computaci\'on,\\

\vspace{-0.2cm}

\noindent Universidad de Santiago de Chile\\

\vspace{-0.2cm}

\noindent \textit{e-mail:} andres.navas@usach.cl

\vspace{0.1cm}

\noindent Gast\'on Vergara-Hermosilla

\vspace{-0.2cm}

\noindent  Institut de Math\'ematiques de Bordeaux, \\

\vspace{-0.2cm}

\noindent Universit\'e de Bordeaux, Francia\\

\vspace{-0.2cm}

\noindent \textit{e-mail:} gaston.vergara@u-bordeaux.fr

\end{small}

\end{document}